\documentclass[10pt, conference, letterpaper]{IEEEtran}

\makeatletter
\def\ps@headings{%
\def\@oddhead{\mbox{}\scriptsize\rightmark \hfil \thepage}%
\def\@evenhead{\scriptsize\thepage \hfil \leftmark\mbox{}}%
\def\@oddfoot{}%
\def\@evenfoot{}}
\makeatother 

\usepackage{mathrsfs}
\usepackage{amsfonts}
\usepackage{array}
\usepackage{amssymb}
\usepackage{amsmath}
\usepackage{graphicx}
\usepackage{multirow}
\usepackage{amsthm}

\usepackage[lined,vlined,ruled,commentsnumbered]{algorithm2e}
\usepackage{epstopdf}
\usepackage{stmaryrd}
\usepackage{multirow,url}
\usepackage{color,cite}
\usepackage{bm}
\usepackage{indentfirst}
\usepackage{epsfig}

\usepackage[table]{xcolor}

\usepackage{fancyhdr}

\fancypagestyle{firststyle}
{
   \fancyfoot[C]{\rev{xxx-x-xxxx-xxxx-x/xx/\$31.00 \copyright 2015 IEEE}}
}

\IEEEoverridecommandlockouts

\newtheorem{assumption}{Assumption}

\newtheorem{lemma}{Lemma}

\newtheorem{proposition}{Proposition}
\newtheorem{definition}{Definition}
\newtheorem{theorem}{Theorem}

\newcommand{\footnotesc}[1]{\footnote{#1}}

\theoremstyle{plain}

\ifodd 1212
\newcommand{\rev}[1]{{\color{blue}#1}}
\newcommand{\revh}[1]{{\color{blue}#1}}

\newcommand{\comg}[1]{\textbf{\color{green} (COMMENT: #1)}}
\newcommand{\response}[1]{\textbf{\color{green} (RESPONSE: #1)}}
\else
\newcommand{\rev}[1]{{#1}}
\newcommand{\revh}[1]{#1}

\newcommand{\comg}[1]{}
\newcommand{\response}[1]{}
\fi

\def\Ex{\mathrm{E}}

\def\N{N}            
\def\Nset{\mathcal{\N}}   
\def\n{n}             

\def\l{s}             

\def\B{\textsf{b}}
\def\LS{\textsf{l}}
\def\A{\textsf{a}}

\def\K{K}                
\def\Kset{\mathcal{\K}}   
\def\k{k}

\def\ch{{channel}}
\def\chs{{channels}}

\def\tvch{{TV \ch}}

\def\db{{database}}

\def\lh{{licensee}}

\def\eu{{user}}
\def\eus{{users}}
\def\Eu{{User}}

\def\Nkset{\Nset_{\k}}

\def\RB{B}  
\def\RL{L}  
\def\RA{A}  

\def\Prob{\eta}      
\def\BProb{\boldsymbol{\eta}}      
\def\Proba{\Prob_{\textsc{a}}}      
\def\Probb{\Prob_{\textsc{b}}}      
\def\Probl{\Prob_{\textsc{l}}}      

\def\fx{f}  
\def\gy{g}  

\def\th{\theta}      
\def\thlb{\th_{\textsc{LB}}}  
\def\thab{\th_{\textsc{AB}}}  
\def\thla{\th_{\textsc{LA}}}  

\def\p{\pi}           

\def\pa{\p_\textsc{a}}           
\def\pl{\p_\textsc{l}}           

\def\w{w}  

\def\U{\Pi}
\def\Ur{\widetilde{\U}}
\def\Ueu{\U^{\textsc{eu}}}
\def\Udb{\U^{\textsc{db}}}
\def\Usl{\U^{\textsc{sl}}}

\def\Udbo{\U^{\textsc{db}}_{0}}
\def\Uslo{\U^{\textsc{sl}}_{0}}

\def\Udbrs{\U^{\textsc{db}}_{\textsc{(i)}}}
\def\Uslrs{\U^{\textsc{sl}}_{\textsc{(i)}}}
\def\Urdbrs{\Ur^{\textsc{db}}_{\textsc{(i)}}}
\def\Urslrs{\Ur^{\textsc{sl}}_{\textsc{(i)}}}

\def\eq{\triangleq}

\def\R{R}
\def\Ra{\R_{\A}}

\def\InfTV{v}         
\def\InfOut{o}        
\def\InfEU{w}         
\def\InfKnown{\bar{z}}      
\def\InfKnownMin{ \InfKnown_{\textsc{min}} }      
\def\InfUnknown{\hat{z}}    
\def\InfTot{z}        
\def\InfTotA{ \InfTot_{(\A)} }        

\def\pl{\p_\textsc{l}}     

\def\t{t}

\def\ut{U}
\def\rt{\mathcal{R}}

\begin{document}

\title{HySIM: A Hybrid Spectrum and Information Market for TV White Space Networks\vspace{-8mm}
}

\author{Yuan~Luo,
        Lin~Gao,
        and~Jianwei~Huang
\thanks{\rev{This work is supported by ...}}
\thanks{The authors are with Dept.~of Information Engineering, The Chinese University of Hong Kong, HK,
Email: \{ly011, lgao, jwhuang\}@ie.cuhk.edu.hk.}
}

\addtolength{\abovedisplayskip}{-1mm}
\addtolength{\belowdisplayskip}{-1mm}

\maketitle

\begin{abstract}
\rev{
We propose a hybrid spectrum and information market for a database-assisted TV white space network, where the geo-location database serves as both a spectrum market platform
and an information market platform.
We study the interactions among the database operator, the spectrum licensee, and unlicensed users systematically, using a three-layer hierarchical model.
In Layer I, the database and the licensee negotiate the commission fee that the licensee pays for using the spectrum market platform.
In Layer II, the database and the licensee compete for selling information or channels to unlicensed users.
In Layer III, unlicensed users determine whether they should buy the exclusive usage right of licensed channels from the licensee, or the   information regarding unlicensed channels from the database.
Analyzing such a three-layer model is challenging due to the co-existence of both positive and negative network externalities in the information market.
We characterize how the network externalities affect the equilibrium behaviours of all parties involved.
Our numerical results show that the proposed hybrid market can improve the network profit up to $87\%$, compared with a pure information market. Meanwhile, the achieved network profit is very close to the coordinated benchmark solution (the gap is less than $4\%$ in our simulation).
}
\end{abstract}

\IEEEpeerreviewmaketitle



\section{Introduction}\label{sec:intro}



\subsection{Background}

With the explosive growth of mobile smartphones and bandwidth-hunger wireless applications, radio spectrum has become increasingly scarce \cite{sen2013sdp}.
The UHF/VHF frequency band originally assigned for broadcast television services (hereafter called TV channels) has been viewed as a promising spectrum opportunity for supporting new wireless broadband services.
First, in many places there are many vacant (unused) TV channels (i.e., those unlicensed to any TV licensee),
often called \emph{TV white spaces} \cite{Microsoft2010report, Ofcom2012tvws, federal2012third}, which can be used for supporting unlicensed non-TV wireless services.
Second, even the licensed TV channels (i.e., allocated to certain TV licensee) may be under a low utilization in most time \cite{cognitive}, and hence can be opportunistically reused by unlicensed non-TV wireless services with the permissions of licensees.
 

To effectively exploit the TV white spaces while not harming the interests of licensed devices (TVs),
the industry has started to adopt a \emph{database-assisted} TV white space network architecture  \cite{Ofcom2010geo, Google, Microsoft, SpectrumBridge}.
In this architecture, unlicensed  devices obtain the list of available unlicensed TV channels via querying a certified white space \emph{geo-location} database, which periodically updates information based on a repository of licensees.
Meanwhile, the FCC also allows the spectrum licensees to temporarily   lease their licensed channels to unlicensed devices through, for example,  auction \cite{federal2004leasing}.
Such spectrum trading (leasing) requires a market platform, and the geo-location database can potentially serve as such a platform (e.g., \emph{SpecEx} \cite{SpectrumBridgeCommericial2}) due to its proximity to both spectrum licensees and unlicensed devices.
This means that the geo-location database can facilitate the unlicensed spectrum access to both unlicensed and licensed TV channels.

Recently, researchers have proposed several business and marketing models related to the database-assisted spectrum sharing, which can be categorized into two classes: \emph{Spectrum Market} and \emph{Information Market}.
The first class ({spectrum market}   \cite{luo2012,feng2013database,luo2013}) mainly deals with the trading of
licensed TV channels.
The key idea is to let spectrum licensees temporarily lease their {under-utilized} licensed TV channels to unlicensed users for some additional revenue.
The database serves as a market platform facilitating this trading process.\footnote{For example, it acts as a spectrum broker or agent, purchases spectrum from licensees and then resells the purchased spectrum to unlicensed users.}
A commercial example of such a \emph{database-provided} spectrum market platform is \emph{SpecEx} \cite{SpectrumBridgeCommericial2}, operated by SpectrumBridge.\footnote{\rev{The (secondary) spectrum market has been extensively used in dynamic spectrum access networks (see, e.g., \cite{zhou2008,zhou2009,newadd-0,newadd-1,newadd-2}), where auction, contract, and pricing are commonly used theoretic models. 
The auction and contract models usually focus on the information asymmetry. 
In this work, we mainly focus on the interplay between the spectrum market and information. Hence, we will consider the basic pricing model for the spectrum market.}} 

The second class ({information market}) has been recently introduced by Luo \emph{et al.} for the unlicensed TV channels (i.e., TV white spaces)\cite{luo2014wiopt, luo2014SDP}.
In their models, the geo-location database sells the advanced information regarding the quality of unlicensed TV channels, instead of channels, to the unlicensed users for profit.
The key motivation is that the database knows more information regarding TV white spaces than unlicensed users,\footnote{For example, based on the knowledge about the network infrastructures of TV licensees and their licensed channels, the database can predict the average interference (from licensed devices) on each TV channel at each location.} and hence it can provide information that helps unlicensed users improve their performances.
A practical example of information market is \emph{White Space Plus} \cite{SpectrumBridgeCommericial}, again operated by SpectrumBridge.

In practice, \emph{both the licensed TV channels and unlicensed TV white spaces co-exist at a particular location}.
Some users may prefer to lease the licensed TV channels from licensees for the exclusive usage,
while other users may prefer to share the free unlicensed TV white spaces with others.
Hence, a joint formulation and optimization of both spectrum market and information market is highly desirable.
However, none of the existing work \cite{luo2012,feng2013database,luo2013,luo2014wiopt, luo2014SDP} looked at the interaction between spectrum market and information market.
This motivates our study of a \emph{hybrid} spectrum and information market for the database in TV white space networks.

\begin{figure}
\centering
  \includegraphics[width=3.1in]{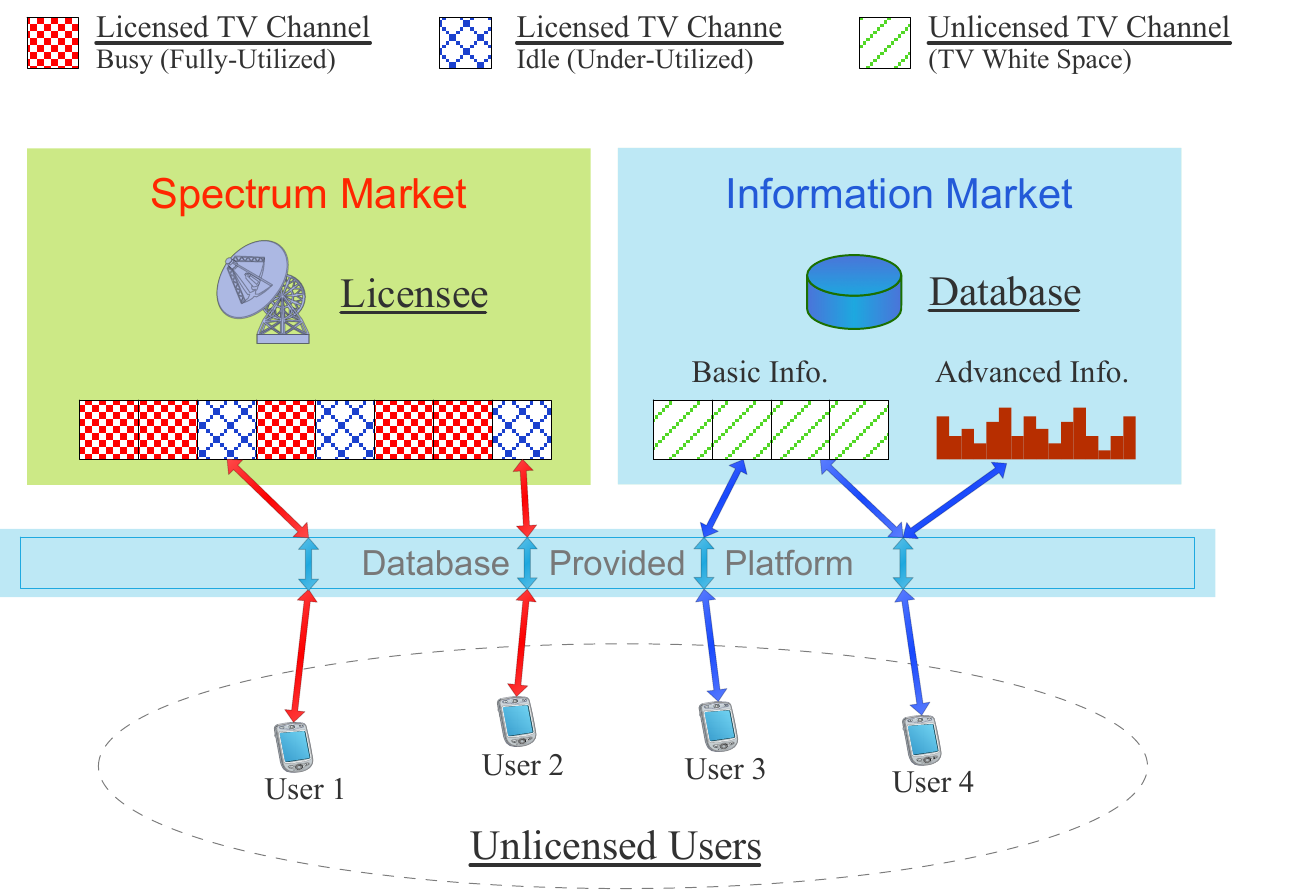}
  \caption{Database-Provided Hybrid Spectrum and Informaiton Market.}\label{fig:model}
  \vspace{-3mm}
\end{figure}

\subsection{Contributions}

In this paper, we model and study a \underline{Hy}brid \underline{S}pectrum and \underline{I}nformation \underline{M}arket (HySIM) for a database-assisted TV white space network,
in which the geo-location database serves as (i) a \emph{spectrum market} platform for the trading of (under-utilized) licensed TV channels between spectrum licensees and unlicensed users,
and (ii) an \emph{information market} platform for the trading of advanced information (regarding the unlicensed TV white spaces) between the database itself and unlicensed users.
Unlicensed users can choose to lease the licensed TV channels from licensees (via the database) for the exclusive usage,
or to share the free TV white spaces with others.
In the latter case, users can further decide whether to purchase the advanced information regarding these TV white spaces from the database to enhance the performance.

Figure \ref{fig:model} illustrates such a database-provided HySIM framework. Unlicensed users 1 and 2 lease the licensed TV channels from the spectrum licensee (via the database-provided platform), and users 3 and 4 share the free unlicensed TV white spaces with others. 
User 4 further purchases the advanced information to improve its performance.


In order to thoroughly understand the user behaving, the market evolving, and the equilibrium in such a hybrid market, we formulate the interactions among the geo-location database (operator), the spectrum licensee, and unlicensed users as a three-layer hierarchical model:

\subsubsection{Layer I: Commission Negotiation (in Section \ref{sec:layer1})}
In the first layer, the database and the licensee negotiate the commission fee that the spectrum licensee needs to pay for using the spectrum market platform. 
Specifically, the database, as the spectrum market platform, helps the spectrum licensee to display, advertise, and sell the under-utilized TV channels to unlicensed users.
Accordingly, it takes some \emph{commission fee} from each successful transaction between the spectrum licensee and unlicensed users.
\rev{In this work, we consider the \emph{revenue sharing} scheme (RSS), where the licensee shares a fixed percentage of revenue with the database,\footnote{\rev{Another commonly-used commission scheme is the so-called \emph{wholesale pricing} scheme (WPS), where the database charges the licensee a fixed price for each successful transaction, regardless of the exact revenue of the licensee. We will study the problem under WPS in our future work.}} 
and study the RSS negotiation using the Nash bargaining theory \cite{harsanyi1977bargaining}.}~~~~~~~~~~~~~~~~~

\subsubsection{{Layer II: Price Competition Game} (in Section \ref{sec:layer2})}
In the second layer, the database and the spectrum licensee compete with each other for selling information or channels to unlicensed users.
The spectrum licensee decides the price of the licensed TV channels,
and the database decides the price of the advanced information (regarding the unlicensed TV channels).
We analyze the equilibrium of such a price competition game using the supermodular game theory \cite{topkis1998supermodular}.~~~~~~~~ 

\subsubsection{{Layer III: User Behaving and Market Dynamics (in Section \ref{sec:user_subscript})}}
In the third layer, unlicensed users decide the best purchasing decisions, given the database's information price and the licensee's channel price.
Note that the users' best purchasing behaviours dynamically change due to the \emph{negative} and \emph{positive} network externalities of the information market (see Section \ref{sec:model-externality} for details).
We will show how the market dynamically evolves according to the users' best choices, and what the \emph{market equilibrium} point is.

In summary, we list the main contributions as follows. 

\begin{itemize}

\item
\emph{Novelty and Practical Significance:}
To the best of our knowledge, this is the first paper that proposes and studies a hybrid spectrum and information market for promoting the unlicensed spectrum access to both licensed and unlicensed TV channels.

\item
\emph{Modeling and Solution Techniques:}
We formulate the interactions as a three-layer hierarchical model, and analyze the model by backward induction, using market equilibrium theory, supermodular game theory, and Nash bargaining theory, respectively. 

\item
\emph{Performance Evaluations:}
{Our numerical results show that the proposed hybrid market can bring up to $87\%$ network profit gain, compared with a pure information market.
The gap between our achieved network profit and the coordinated benchmark is less than $4\%$.
}

\end{itemize}




\section{System Model}\label{sec:model}

We consider a database-assisted TV white space network with a \emph{geo-location {\db}} and a set of \emph{unlicensed users} (devices) in a particular region (e.g., a city). 
Unlicensed users can use the unlicensed TV channels (i.e., TV white spaces) freely in a shared manner (e.g., using CDMA or CSMA). 
Meanwhile, there is a \emph{spectrum licensee}, who owns the licensed channels and wants to lease the under-utilized channels to unlicensed users for additional revenue.\footnote{In case there are multiple spectrum licensees, we assume that they are coordinated by the single representative. We will leave the case with multiple competitive spectrum licensees to a future work.}
Different from the unlicensed TV white spaces, the licensed TV channels can be used by unlicensed users in an exclusive manner (with the permission of the licensee).
Therefore, users can enjoy a better performance (e.g., a higher data rate or a lower interference) on the licensed channels.
For convenience, let $\pl \geq 0$ denote the (licensed) {channel price} set by the spectrum licensee.

\subsection{Geo-location Database}

\subsubsection{Basic Service}
According to the regulation policy (e.g., \cite{federal2012third}), it is mandatory for a geo-location white space {\db} to provide the following information for any unlicensed {\eu}: (i) the list of TV white spaces (i.e., unlicensed TV channels), (ii) the transmission constraint (e.g., maximum transmission power) on each channel in the list, and (iii) other optional requirements.
The database needs to provide this \emph{basic (information) service} free of charge for any unlicensed user.

\subsubsection{Advanced Service}
Beyond the basic information, the {\db} can also provide certain advanced information regarding the quality of TV channels (as SpectrumBridge did in \cite{SpectrumBridgeCommericial}), which we call the \emph{advanced (information) service}, as long as it does not conflict with the free basic service.
Such an advanced information can be rather general, and a typical example is ``the interference level on each channel'' used in \cite{luo2014wiopt, luo2014SDP}.
With the advanced information, the {\eu} is able to choose a channel with the  highest quality (e.g., with the lowest interference level).
Hence, the {\db} can \emph{sell} this advanced information to users for profit.
This leads to an \emph{information market}.
For convenience, let $\pa \geq 0$ denote the (advanced) {information price} of the database.~~~~~~~~~~

\subsubsection{Leasing Service}
As mentioned previously, the geo-location database can also serve as a spectrum market \emph{platform} for the trading of licensed channels between the spectrum licensee and unlicensed users, which we call the \emph{leasing service}.
\rev{By doing this, the spectrum  licensee shares a fixed percentage $\delta \in [0,1]$ of  revenue with the database, which we called the revenue share commission scheme (RSS).} 

 \subsection{Unlicensed User}

Unlicensed users can choose either to purchase the licensed channel from the licensee for the exclusive usage, or to share the free unlicensed TV white spaces with others (with or without advanced information).
We assume that all licensed and unlicensed TV channels have the same bandwidth (e.g., 6MHz in the USA), and each user only needs one channel (either licensed or unlicensed) at a particular time.
Formally, we denote $\l \in \{\B, \A, \LS\}$ as the  \emph{strategy} of a user, where

\vspace{-1mm}
(i)
$\l = \B$:
Choose the basic service (i.e., share TV white spaces with others, without the advanced information);

\vspace{-1mm}
(ii)
$\l = \A$:
Choose the advanced service (i.e., share TV white spaces with others, with the advanced information).

\vspace{-1mm}
(iii)
$\l = \LS$:
Choose the leasing service (i.e., lease the licensed channel from the licensee for the exclusive usage).

We further denote $\RB$, $\RA$, and $\RL$ as the expected \emph{utility} that a user can achieve from choosing the basic service ($\l = \B$), the advanced service ($\l = \A$), and the leasing service ($\l = \LS$), respectively.
The \emph{payoff} of a {\eu} is defined as the difference between the achieved utility and the service cost (i.e., the  {information price} when choosing the advanced service, or the  {leasing price} if choosing the leasing service).
\rev{
Let $\th$ denote the \eu's evaluation for the achieved utility.
}
Then, the payoff of a user with an evaluation factor $\th$ can be written as 
\begin{equation}\label{eq:utility-basic}
\textstyle
\Ueu_{\th} = \left\{
  \begin{aligned}
  &\textstyle  \th \cdot \RB ,      &&  \ \text{if} ~ \l = \B, \\
  &\textstyle  \th \cdot \RA   -  \pa , &&  \  \text{if} ~ \l  = \A, \\
  &\textstyle  \th \cdot \RL - \pl ,      &&  \ \text{if} ~ \l = \LS.
   \end{aligned}
\right.
\end{equation}
Each user is rational and will choose a strategy $\l \in \{\B, \A, \LS\}$ that maximizes its payoff.
Note that different {\eus} may have different values of $\th$ (e.g., depending on application types), hence have different choices. That is, {\eus} are heterogeneous in term of $\th$.
For convenience, we assume that $\th$ is uniformly distributed in $[0,1]$ for all \eus\footnote{{This assumption is commonly used in the existing literature. Relaxing to more general distributions often does not change the main insights \cite{manshaei2008evolution,shetty2010congestion}.}}.



Let $\Probb$, $\Proba$, and $\Probl$ denote the fraction of {\eus} choosing the basic service, the advanced service, and the leasing service, respectively.
For convenience, we refer to $\Probb$, $\Proba$, and $\Probl$ as the \emph{market shares} of the basic service, the advanced service, and the leasing service, respectively.
Obviously, $\Probb, \Proba, \Probl \geq 0$ and $\Probb + \Proba + \Probl = 1$.
Then, the {normalized} \emph{payoffs} (profits) of the spectrum licensee and the database are, respectively,
\begin{equation}\label{eq:u1}
\left\{
\begin{aligned}
\Usl \eq \Uslrs &= \pl   \Probl   (1 - \delta),
\\
\Udb \eq \Udbrs &= \pa   \Proba + \pl   \Probl   \delta.
\end{aligned}
\right.
\end{equation}

 \subsection{Positive and Negative Network Externalities}\label{sec:model-externality}
 
There are two types of network externalities coexisting in the information market: (i) \emph{negative externality}, which corresponds to the increasing level of congestion and degradation of user performance due to more users sharing the same TV white space, and (ii) \emph{positive externality}, which is due to the quality increase of the (advanced) information when more users purchasing the information.
Next we analytically quantify these two network externalities.

We first have the following intuitive observations for a user's expected utilities of three strategy choices:

\vspace{-1mm}
$\bullet$
\emph{$\RL$ is a constant and independent of $\Proba $, $\Probb$, and $ \Probl $.}
This is because a user uses the licensed channels in an exclusive manner, hence its performance (on licensed channels) does not depend on the activities of others.

\vspace{-1mm}
$\bullet$
\emph{$\RB$ is non-increasing in $\Proba + \Probb$} (the total fraction of users using TV white space) due to the congestion effect.
This is because more users using TV white spaces (in a shared manner) will increase the level of congestion on these channel, hence reduce the performance of each user. 
 
 \vspace{-1mm}
$\bullet$
\emph{$\RA$ is non-increasing in $\Proba + \Probb$, due to the congestion effect} (similar as $\RB$).
This is referred to as the \textbf{negative network externality} of the information market.

\vspace{-1mm}
$\bullet$
\emph{$\RA$ is non-decreasing in $\Proba$, given a fixed value of $\Proba + \Probb $}.
This is because more users purchasing the advanced information will increase the quality of the information.
This is referred to as the \textbf{positive network externality}.

For convenience, we write $\RB$ as a non-increasing function $\fx(\cdot)$ of  $ \Proba + \Probb $ (or equivalently, $1 - \Probl $), i.e.,
$$
\textstyle
\RB \triangleq \fx(\Proba + \Probb),
$$
and write $\RA$ as the combination of a
non-increasing function $\fx(\cdot)$  of $ \Proba + \Probb $ and a non-decreasing function  $\gy(\cdot)$  of $\Proba $, i.e.,
$$
\RA \triangleq \fx(\Proba + \Probb) + \gy(\Proba).
$$
Note that $\fx (\cdot)$ reflects the congestion effect in the information marekt, and is identical in $\RB $ and  $\RA$ (as users experience the same congestion effect in both basic and advanced services), and  $\gy(\cdot)$  reflects the performance gain induced by the advanced information, i.e., the value of advanced information.

Since there is no congestion on the licensed channels, it is reasonable to assume that $\RL>\RA $ and $\RL>\RB $.
We can further assume that $\RA>\RB$, that is, the additional gain $\gy(\Proba) $ from the advanced information is positive.\footnote{{Note that if we assume  $\RL < \RB$, then users will never choose the leasing service even with a zero channel price $\pl$.
In this case, our model degenerates to the pure information market, similar as that in \cite{luo2014wiopt}.
Moreover, if $\RA = \RB$, then users will never choose the advanced service even with a zero information price $\pa$.
In this case, our model degenerates to a monopoly spectrum market (where the licensee is the monopolist).
In this sense, our hybrid market model generalizes both the pure spectrum market and pure information market.}} 
To facilitate the later analysis, we further introduce the following assumptions on functions $\fx (\cdot)$ and $\gy (\cdot)$.
\begin{assumption}\label{assum:congestion}
  $\fx(\cdot)$ is non-negative, non-increasing, convex, and continuously differentiable.
\end{assumption}
\begin{assumption}\label{assum:positive}
  $\gy(\cdot)$ is non-negative, non-decreasing, concave, and continuously differentiable.
\end{assumption}

\rev{
The non-increasing and convexity assumption of $\fx(\cdot)$ reflects the increasing of marginal performance degradation under congestion, and is widely used in wireless networks with congestion effect (see, e.g., \cite{shetty2010congestion,johari2010congestion} and references therein).
The non-decreasing and concavity assumption of $\gy(\cdot)$ reflects the diminishing of marginal performance improvement induced by the advanced information.
In this work, we use the generic functions $\fx(\cdot)$ and $\gy(\cdot)$, which can generalize many practical scenarios with the explicit advanced information definition (e.g., those proposed by Luo \emph{et al.} in \cite{luo2014wiopt, luo2014SDP}, where the advanced information is the interference level on each channel). 
We provide more detailed discussion about generic functions and practical scenarios in the Appendix of \cite{report}. 
}

\begin{figure}[t]
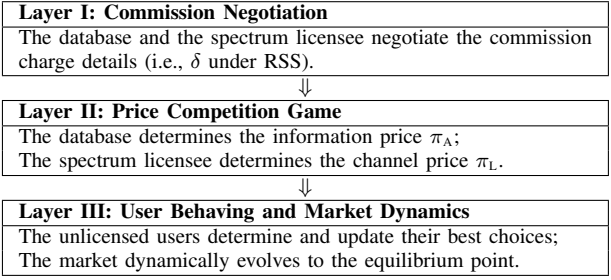

\centering
\footnotesize
\begin{tabular}{|m{3in}|}
\hline
\textbf{\revh{Layer I}: Commission Negotiation}
\\
\hline
The database and the spectrum licensee negotiate the commission charge details (i.e.,
$\delta $ under RSS).
\\
\hline
\multicolumn{1}{c}{$\Downarrow$} \\
\hline
\textbf{\revh{Layer II}: Price Competition Game}
\\
\hline
The database determines the information price $\pa$;
\\
The spectrum licensee determines the channel price $\pl$.
\\
\hline
\multicolumn{1}{c}{$\Downarrow$} \\
\hline
\textbf{\revh{Layer III}: User Behaving and Market Dynamics}
\\
\hline
The unlicensed users determine and update their best choices;
\\
The market dynamically evolves to the equilibrium point.
\\
\hline
\end{tabular}
\caption{Three-layer Hierarchical Interaction Model}
\label{fig:layer}
\vspace{-3mm}
\end{figure}

\subsection{Three-Layer Interaction Model}

%

Based on the above discussion, a hybrid spectrum and information market involves the interactions among the geo-location database, the spectrum licensee, and the unlicensed users.
Hence, we formulate the interactions as a three-layer hierarchical model illustrated in Figure \ref{fig:layer}.

Specifically, in Layer I, the database and the spectrum  licensee negotiate the commission charge details (regarding the spectrum market platform), i.e., the revenue sharing factor $\delta \in[0,1] $.
In Layer II, the database and the spectrum licensee compete with each other to attract unlicensed users. The database determines the price $\pa$ of the advanced information, and the spectrum licensee determines the price $\pl$ of the licensed channel.
In Layer III, the unlicensed users determine their best choices, and dynamically update their choices based on the current market shares. Accordingly, the market dynamically involves and finally reaches the equilibrium point.

In the following sections, we will analyze this three-layer interaction model systematically using backward induction.

\section{Layer III -- User Behavior and Market Equilibrium}
\label{sec:user_subscript}

In this section, we study the user behavior and market dynamics in Layer III, given the database's information price $\pa$ and the licensee's channel price $\pl$ (in Layer II).
In the following, we first discuss the {\eu}'s best choice, and show how the user behavior dynamically evolves, and how the market converges to an equilibrium point.

 \subsection{{\Eu}'s Best Strategy}\label{sec:user-best}

Now we study the best strategy of   {\eus}, given the prices $\{\pl, \pa\}$ and the initial market state $\{\Probl^0, \Proba^0, \Probb^0\}$ where $\Probb^0 + \Proba^0 + \Probl^0 = 1$.
Notice that each {\eu} will choose a strategy that maximizes its payoff defined in (\ref{eq:utility-basic}).
Hence, for a type-$\th$ {\eu}, its best strategy is
\footnote{{Here, ``iff'' stands for ``if and only if''.
Note that we omit the case of $\th \cdot \RL -  \pl = \max\{ \th \cdot \RA  -  \pa ,\ \th \cdot \RB\}$, $\th \cdot \RA -  \pa = \max\{ \th \cdot \RL  -  \pl ,\ \th \cdot \RB\}$, and $\th \cdot \RB  = \max\{ \th \cdot \RL - \pl, \th \cdot \RA  -  \pa \}$, which are negligible (i.e., occurring with zero probability) due to the continuous distribution of $\th$.}}~~~~
\begin{equation}
\left\{
\begin{aligned}
\l_{\th}^* = \LS, & \mbox{~~~~iff~~} \th \cdot \RL -  \pl > \max\{ \th \cdot \RA  -  \pa ,\ \th \cdot \RB\}
\\
\l_{\th}^* = \A, & \mbox{~~~~iff~~} \th \cdot \RA  -  \pa >  \max\{ \th \cdot \RL  -  \pl ,\
 \th \cdot \RB\}
\\
\l_{\th}^* = \B, & \mbox{~~~~iff~~} \th \cdot \RB > \max\{ \th \cdot \RL  -  \pl, \ \th \cdot \RA -  \pa \}
\end{aligned}
\right.
\end{equation}
where $\RB =  \fx(1- \Probl^0)$, and $\RA = \fx(1- \Probl^0) + \gy(\Proba^0)$.

To better illustrate the above best strategy, we introduce the following notations:
\begin{equation*}\label{eq:p-thres}
\textstyle
\thlb \eq \frac{ \pl}{ \RL- \RB },
~~~~
\thab \eq \frac{ \pa}{ \RA- \RB },
~~~~
\thla \eq \frac{\pl-\pa}{\RL - \RA}.
\end{equation*}
Intuitively, $\thlb$ denotes the smallest $\th$ such that a type-$\th$ \eu~prefers the leasing service than the basic service; $\thab$ denotes the smallest $\th$ such that a type-$\th$ \eu~prefers the advanced service than the basic service; and $\thla$ denotes the smallest $\th$ such that a type-$\th$ \eu~prefers the leasing service than the advanced service.
Notice that $\RA$ and $\RB$ are functions of initial market shares $\{\Probl^0, \Proba^0, \Probb^0\}$. Hence, $\thlb $, $\thab$, and $\thla$ are also functions of $\{\Probl^0, \Proba^0, \Probb^0\}$.

\begin{figure}
\centering
  \includegraphics[width=2.6in]{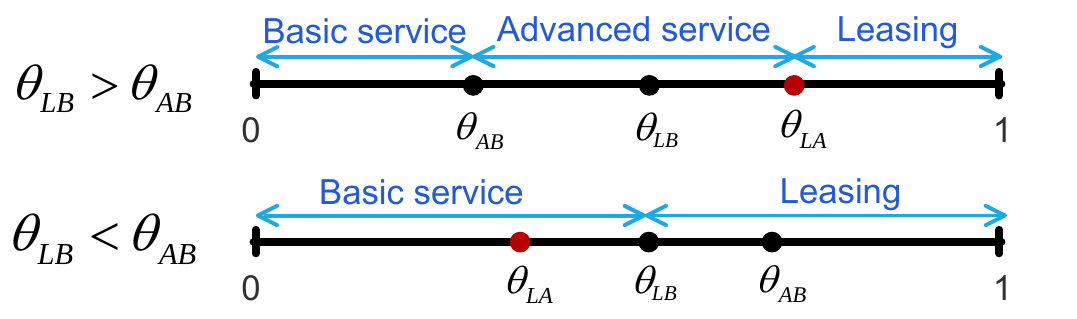}
  \vspace{-2mm}
  \caption{Illustration of $\thlb$, $\thab$, and $\thla$.}\label{fig:threshold}
  \vspace{-4mm}
\end{figure}


\rev{
Figure \ref{fig:threshold} illustrates the relationship of $\thlb$, $\thab$, and $\thla$.
Intuitively, Figure \ref{fig:threshold} implies that the users with a high utility evaluation factor $\th$ are more willing to choose the leasing service in order to achieve a large utility.
The users with a low utility evaluation factor $\th$ are more willing to choose the basic service so that they will pay zero service cost.
The users with a middle utility evaluation factor $\th$ are willing to choose the advanced service, in order to achieve a relatively large utility with a relatively low service cost.
Notice that when the information price $\pa$ is high or the information value (i.e., $\RA-\RB$) is low, we could have $\thlb < \thab$, in which no users will choose the advanced service (as illustrated in the lower subfigure of Figure \ref{fig:threshold}).
}


Next we characterize the new market shares (called the \emph{derived market shares}) resulting from the users' best choices mentioned above.
Such derived market shares are important for analyzing the user behavior dynamics and market evolutions in the next subsection.
Recall that $\th$ is uniformly distributed in $[0,1]$.
Then, given any initial market shares $\{\Probl^0, \Proba^0 \}$, the newly derived market shares $\{ \Probl,\Proba \}$ are
\begin{itemize}
\item
If $\thlb > \thab$, then $\Probl = 1 -  \thla$ and $\Proba  = \thla - \thab$;
\item
If $\thlb \leq \thab$, then $\Probl = 1 - \thlb$ and $\Proba  = 0$.
\end{itemize}

Formally, we have the following derived market shares.
\begin{lemma}\label{lemma:market-share}
Given any initial market shares $\Probl^0$ and $\Proba^0$, the derived market shares  $\Probl$ and $\Proba$ are given by
\begin{equation}\label{eq:user-prob-1}
\textstyle
\left\{
\begin{aligned}
\Probl  & \textstyle
= \max \big\{ 1 - \max\{ \thla, \thlb \} ,\ 0  \big\},\\
\Proba  & \textstyle
= \max \big\{ \min\{ \thla ,1 \} - \thab ,\ 0  \big\}.
\end{aligned}
\right.
\end{equation}
\end{lemma}

The results in Lemma \ref{lemma:market-share} assume that all \eus~update the best strategies once and simultaneously.
Since $\thlb$, $\thab$, and $\thla$ are functions of initial market shares $\{\Probl^0, \Proba^0\}$, the derived market shares $\{\Probl ,\Proba \}$  are also functions of  $\{\Probl^0,\Proba^0\}$, and hence can be written as  $\Probl(\Probl^0, \Proba^0)$ and $\Proba(\Probl^0, \Proba^0)$.

\subsection{{Market Dynamics and Equilibrium}}

When the market shares change,
the users' payoffs (on the advanced service and basic service) change accordingly, as $\RA$ and $\RB$ change.
As a result, users will update their best strategies continuously, hence the market shares will evolve dynamically, until reaching a {stable} point (called \emph{market equilibrium}).
In this subsection, we will study such a market dynamics and equilibrium, given the prices $\{ \pl, \pa \}$.~~~~~~~~~~~~~~~~~~~

For convenience, we introduce a virtual time-discrete system with slots $\t=1,2,\ldots$, where {\eu}s change their decisions at the beginning of every slot, based on the derived market shares in the previous slot.
Let $( \Probl^{\t}, \Proba^{\t} )$ denote the market shares derived at the end of slot $t$ (which serve as the initial market shares in the next slot $t+1$).
We further denote $\triangle \Probl  $ and $\triangle \Proba $ as the changes (dynamics) of market shares between two successive time slots, e.g., $\t$ and $\t+1$, that is,
\begin{equation}\label{eq:user-prob-diff}
\begin{aligned}
\textstyle
\triangle \Probl(\Probl^{\t}, \Proba^{\t}) &= \Probl^{\t+1} - \Probl^{\t }, ~
\triangle \Proba(\Probl^{\t}, \Proba^{\t}) &= \Proba^{\t+1} - \Proba^{\t },
\end{aligned}
\end{equation}
where $( \Probl^{\t+1}, \Proba^{\t+1} )$ are the derived market share in slot $\t+1$, which can be computed by Lemma \ref{lemma:market-share}.
Obviously, if both $\triangle \Probl $ and $\triangle \Proba $ are zero in  a slot $\t+1$, i.e., $\Probl^{\t+1} = \Probl^{\t} $ and $\Proba^{\t+1} = \Proba^{\t} $, then users will no longer change their strategies in the future. This implies that the market achieves a stable state, which we call the \emph{market equilibrium}.
Formally,

\begin{definition}[Market Equilibrium]\label{def:stable-pt}
A pair of market shares $\BProb^{*} = \{ \Probl^{*}, \Proba^{*} \}$ is a market equilibrium, if and only if
\begin{equation} \label{eq:market_equilibrium}
\textstyle
\triangle \Probl (\Probl^*, \Proba^*) = 0, \mbox{~~~and~~~}
\triangle \Proba (\Probl^*, \Proba^*) = 0.
\end{equation}
\end{definition}

Next, we study the existence and uniqueness of the market equilibrium, and further characterize the market equilibrium analytically.
These results are very important for analyzing the price competition game in Layer II (Section \ref{sec:layer2}).
\begin{proposition}[Existence]\label{lemma:existence-eq_pt}
Given any feasible price pair $( \pl, \pa)$,
there exists at least one market equilibrium.
\end{proposition}
\begin{proposition}[Uniqueness]\label{lemma:uniqueness-eq_pt}
Given any feasible price pair $( \pl, \pa)$,
there exists a unique market equilibrium $(\Probl^*, \Proba^*)$, if there exists a tuple $(\Probl, \Proba)$ with $\Probl + \Proba \leq 1$ such that\footnote{{Here, $\gy^{\prime}(\Proba)$ is the first-order derivative of $\gy(\cdot)$ with respect to $\Proba$. Note that $\RA$ is a function of $\Proba$ and $\Probl$, and $\RB$ is a function of $\Probl$.}}
\begin{equation}\label{eq:stable_condition}
\textstyle
{ \frac{ \gy^{\prime}(\Proba) }{ \gy(\Proba)  } \cdot \frac{ \RL - \RB }{ \RL - \RA }  \leq 1.}
\end{equation}
\end{proposition}


\begin{figure}
\centering
  \includegraphics[width=2.8in]{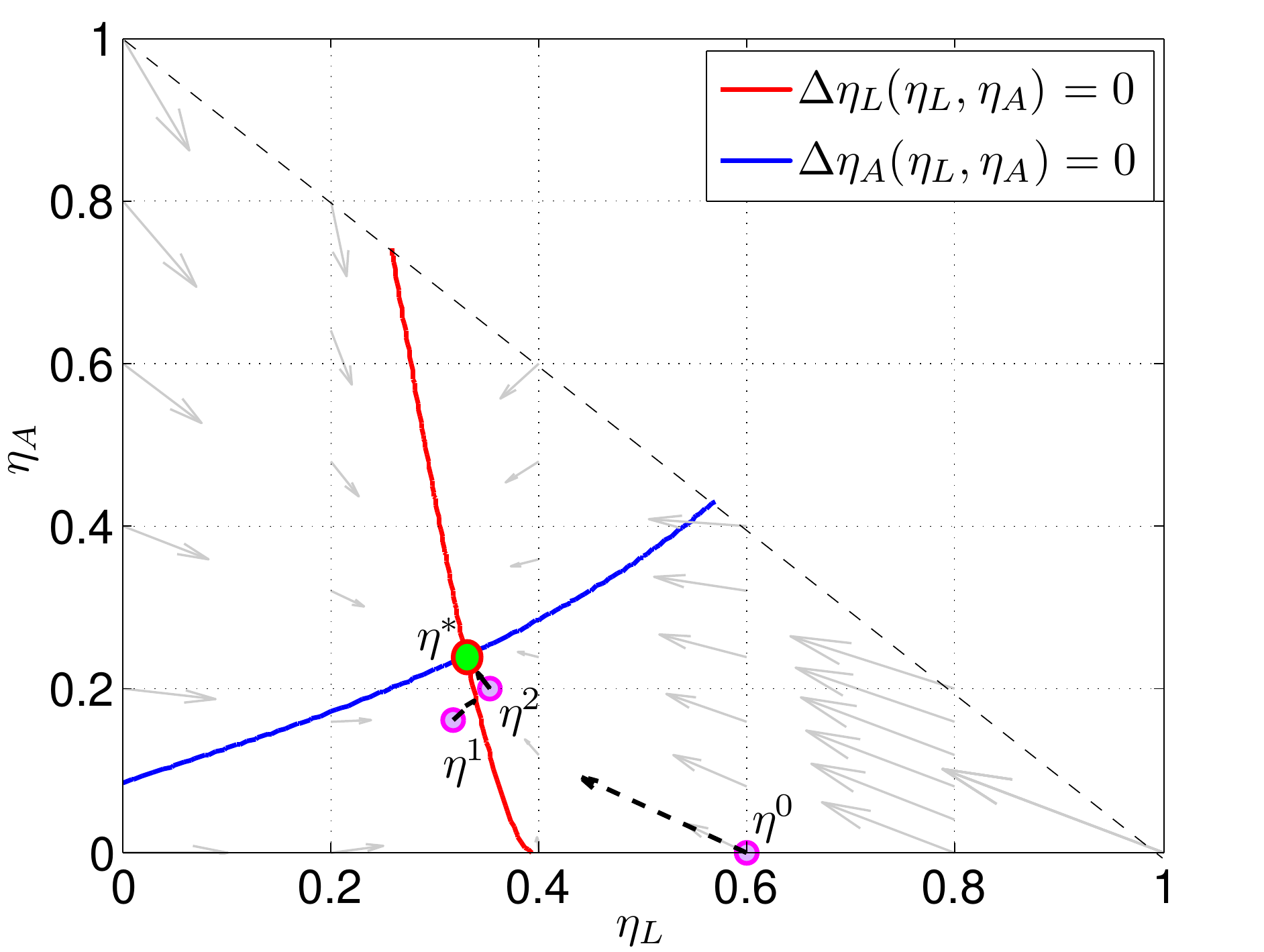}
  \caption{Illustration of Market Dynamics and Market Equilibrium. Red Curve:  $\triangle \Probl(\Probl, \Proba) = 0$; Blue Curve: $\triangle \Proba(\Probl, \Proba) = 0$. The intersection between   blue   and red curve is the market equilibrium.}
  \label{fig:dyna-users}
  \vspace{-6mm}
\end{figure}

A practical implication of (\ref{eq:stable_condition}) is that if the information value $\gy(\Proba)$ (positive network externality) does not increase very fast with $\Proba$, then there exists a unique equilibrium.
Note that the condition (\ref{eq:stable_condition}) is sufficient but not necessary for the uniqueness.
In particular, we observe through numerical simulations that in some cases, the market converges to a unique equilibrium for a wide range of prices under which the condition (\ref{eq:stable_condition}) is violated.
Nevertheless, the sufficient condition in (\ref{eq:stable_condition}) leads to the insight that if the impact of positive network externality is significant, there may exist multiple equilibrium points.
Note that even if there exist multiple equilibrium points,
the market always converges to a unique one of them, given the initial market shares. {Please refer to \cite{report} for more details.}


For a better understanding, we illustrate the dynamics of market shares in Figure \ref{fig:dyna-users}.
The $x$-axis denotes the leasing service's market share-$\Probl$, and the $y$-axis denotes the advanced service's market share-$\Proba$.
Notice that a feasible pair of market shares $\{\Probl, \Proba\}$ satisfies $\Probl + \Proba \leq 1$.
An arrow denotes the dynamics of market shares under a particular initial market shares (at the starting point of the arrow).
For example, {from the initial market shares $\BProb^{0} = \{\Probl = 0.6, \Proba = 0\}$, the  market will evolve to $\BProb^{1} = \{0.32, 0.16\}$, then $\BProb^{2} = \{0.35, 0.2\}$, and eventually converge to the equilibrium point $\BProb^{*} = \{0.33,0.24\}$.}
The red curve denotes the isoline of $\triangle \Probl(\Probl, \Proba) = 0$,  and the blue curve denotes the isoline of $\triangle \Proba(\Probl, \Proba) = 0$.
By Definition \ref{def:stable-pt}, the intersection between the blue curve and the red curve is the market equilibrium point.
In this example, there is a unique market equilibrium point.

Suppose the uniqueness condition (\ref{eq:stable_condition})  is satisfied. We characterize the unique equilibrium by the following theorem.
\begin{theorem}[Market Equilibrium]\label{thrm:stable-eq_pt}
Suppose the uniqueness condition (\ref{eq:stable_condition}) holds.
Then, for any feasible price pair $( \pl, \pa)$, the unique market equilibrium is given by
\begin{itemize}
\item[(a)]
If $ \left.\thlb(\Probl, \Proba)\right|_{\Probl = 0} \leq \left.\thab(\Probl, \Proba)\right|_{\Proba = 0}$, then there is a unique market equilibrium  $\BProb^{\dag} = \{ \Probl^{\dag}, \Proba^{\dag} \}$ given by
    \begin{equation}\label{eq:NE-pt-1}
 \textstyle  \Probl^{\dag} = 1 - \thlb(\Probl^{\dag},\Proba^{\dag}) , \mbox{~~~and~~~}
  \Proba^{\dag} = 0;
    \end{equation}

\item[(b)]
If $ \left.\thlb(\Probl, \Proba)\right|_{\Probl = 0} > \left.\thab(\Probl, \Proba)\right|_{\Proba = 0}$, then there is a unique market equilibrium  $\BProb^{*} = \{ \Probl^{*}, \Proba^{*} \}$ given by
    \begin{equation}\label{eq:NE-pt-22}
    \left\{
      \begin{aligned}
      &\textstyle  \Probl^{*} = 1 - \thla(\Probl^{*}, \Proba^{*}) , \\
      &\textstyle  \Proba^{*} = \thla(\Probl^{*}, \Proba^{*}) - \thab(\Probl^{*}, \Proba^{*}).
       \end{aligned}
    \right.
    \end{equation}
\end{itemize}
\end{theorem}

\begin{proof}
First, we obtain the derived market shares by substituting the market shares given in (\ref{eq:NE-pt-1}) or (\ref{eq:NE-pt-22}) into (\ref{eq:user-prob-1}).
Then, we can check the above derived market shares satisfy the equilibrium condition (\ref{eq:market_equilibrium}).
For the detailed proof, please refer to \cite{report}.
\end{proof}




\vspace{-2mm}
\section{Layer II -- Price Competition Game Equilibrium}\label{sec:layer2}

In this section, we study the price competition between the database and the spectrum licensee in Layer II, given the commission negotiation solution in Layer I and based on the market equilibrium prediction in Layer III.
We will analyze the  game equilibrium under the revenue sharing scheme (RSS).
We first define the price competition game (PCG) explicitly.
\begin{itemize}
\item
\emph{Players:} The database and the spectrum licensee;
\item
\emph{Strategies:} The database's strategy is the price  $\pa$ of its advanced information, and the licensee's strategy is the price  $\pl$ of its licensed channels;
\item
\emph{Payoffs:} The payoffs of players are defined in (\ref{eq:u1}) under RSS.
\end{itemize}
For convenience, we write the (unique) market equilibrium $\BProb^{*} = \{ \Probl^{*}, \Proba^{*} \}$ in Layer III as functions of prices $(\pl, \pa)$, i.e., $\Probl^{*} (\pl, \pa)$ and $\Proba^{*} (\pl, \pa)$.
Intuitively, we can interpret $\Probl^{*} (\cdot)$ and $\Proba^{*} (\cdot)$ as the \emph{demand} functions of the licensee and the \db, respectively.


Assume that the licensee shares a fixed percentage $\delta \in [0, 1]$ of revenue with the database. Then, 
by (\ref{eq:u1}), the payoffs of the \lh~and the \db~can be written as:
\begin{equation}\label{eq:sl-profit-dynamic-rv}
\left\{
\begin{aligned}
\Uslrs(\pl , \pa) & = \pl \cdot \Probl^{*}(\pl, \pa) \cdot (1 - \delta),
\\
\Udbrs(\pl , \pa) &= \pa \cdot \Proba^{*}(\pl, \pa) + \pl \cdot \Probl^{*}(\pl, \pa) \cdot \delta.
\end{aligned}
\right.
\end{equation}


\begin{definition}[Nash Equilibrium]\label{def:nash}
A pair of prices $( \pl^{*}, \pa^{*} )$ is called a Nash equilibrium, if
\begin{equation}\label{eq:db-price-dynamic}
\left\{
\begin{aligned}
\textstyle\pl^{*} & = \arg \max_{\pl \geq 0}\ \Uslrs(\pl , \pa^{*}),
\\
\textstyle\pa^{*} & = \arg \max_{\pa \geq 0}\ \Udbrs(\pl^{*} , \pa).
\end{aligned}
\right.
\end{equation}
\end{definition}

It is notable that directly solving the Nash equilibrium is very challenging, due to the difficulty in analytically characterizing the market equilibrium $\{\Probl^{*}(\pl, \pa),\Proba^{*} (\pl, \pa) \}$ under a particular price pair $\{\pl, \pa\}$.
To this end, we transform the original price competition game (PCG) into an
equivalent \emph{market~share~competition game} (MSCG).
The key idea is to view the market share as the strategy of the database or the licensee, and the prices as functions of the market shares.

Specifically, we notice that under the uniqueness condition (\ref{eq:stable_condition}), there is a \emph{one-to-one} correspondence between the market equilibrium $\{\Probl^{*}, \Proba^* \}$ and the prices  $\{\pl, \pa\}$.

In this sense, once the \lh~and the \db~choose the prices $\{\pl, \pa\}$, they have equivalently chosen the market shares $\{\Probl^{*}, \Proba^* \}$.
Hence, we obtain the equivalent {market share competition game}---MSCG, where the strategy of each player is its market share (i.e., $\Probl$ for the licensee and $\Proba$ for the database), and the prices $\{\pl, \pa\}$ are functions of the market shares $\{\Probl , \Proba  \}$.
{Substitute $\thla=\frac{\pl-\pa}{\RL - \RA}$ and $\thab=\frac{ \pa}{ \RA- \RB }$ into (\ref{eq:NE-pt-22}), we can derive the inverse function of (\ref{eq:NE-pt-22}), where prices are functions of market shares, i.e.,\footnote{We omit the trivial case in (\ref{eq:NE-pt-1}), where the database has a zero market share, as this will never the case at the pricing equilibrium of Layer II.}}
\begin{equation}\label{eq:price-market-share-rs}
\textstyle
\left\{
  \begin{aligned}
  \textstyle  \pl(\Probl , \Proba )   = & ( 1 - \Probl ) \cdot \left( \RL - \fx(1-\Probl) - \gy(\Proba) \right)   \\
&  + ( 1 - \Probl - \Proba )\cdot \gy(\Probl) , \\
 \textstyle  \pa (\Probl , \Proba ) =& ( 1 - \Probl - \Proba )\cdot \gy(\Proba).
   \end{aligned}
   \right.
\end{equation}
Accordingly, the payoffs of two players can be written as:
\begin{equation}\label{eq:sl-profit-dynamic-rv-xx}
\left\{
\begin{aligned}
\textstyle
\Urslrs(\Probl , \Proba) & = \pl (\Probl, \Proba) \cdot \Probl \cdot (1 - \delta),
\\
\textstyle
\Urdbrs(\Probl , \Proba) & = \pa (\Probl , \Proba) \cdot \Probl + \pl (\Probl , \Proba) \cdot \Probl \cdot \delta.
\end{aligned}
\right.
\end{equation}
Similarly, a pair of market shares $(\Probl^*, \Proba^*)$ is called a Nash equilibrium of MSCG, if $\Probl^* = \arg \max_{\Probl} \Uslrs(\Probl , \Proba^*) $ and $\Proba^* = \arg \max_{\Proba} \Udbrs(\Probl^*, \Proba) $.

%
%
%
We first show that the equivalence between the original PCG and the above MSCG.
\begin{proposition}[Equivalence]\label{lemma:game_tranform}
If $\{\Probl^{*}, \Proba^*\}$ is a Nash equilibrium of MSCG, then $ \{\pl^*, \pa^*\}$ given by (\ref{eq:price-market-share-rs}) is a Nash equilibrium of the original price competition game PCG.
\end{proposition}
We next show that the MSCG is a supermodular game (with minor strategy transformation), and then derive the Nash equilibrium using the supermodular game theory \cite{topkis1998supermodular}.
\begin{proposition}[Existence]\label{thrm:NE-existence}
The MSCG is a supermodular game with respect to $\Proba$ and $-\Probl$. Hence, there exists at least one Nash equilibrium $(\Probl^*, \Proba^*)$.
\end{proposition}

The following proposition further gives the uniqueness condition of the Nash equilibrium in MSCG.
\begin{proposition}[Uniqueness]\label{thrm:NE-uniquness}
The MSCG has a unique Nash equilibrium $(\Probl^*, \Proba^*)$, if
$$
 \textstyle - \frac{  \partial^2{ \Urslrs({\Prob}_{\textsc{l}} , \Proba) } }{ \partial{ (-{\Prob}_{\textsc{l}}) }^2 } \geq \frac{  \partial^2{ \Urslrs( {\Prob}_{\textsc{l}} , \Proba)  } }{ \partial{ {(-{\Prob}_{\textsc{l}}) } }\partial{ \Proba  } },
~
 \textstyle - \frac{  \partial^2{ \Urdbrs( {\Prob}_{\textsc{l}} , \Proba) } }{ \partial{ (-{\Prob}_{\textsc{l}}) }^2 } \geq \frac{  \partial^2{ \Urdbrs( {\Prob}_{\textsc{l}} , \Proba)  } }{ \partial{ { \Proba } }\partial{ (-{\Prob}_{\textsc{l}})  } }.
$$
\end{proposition}
The above uniqueness conditions are quite general, and follow the standard supermodular game theory.
Next we provide a specific example to illustrate these conditions more intuitively.
Consider the following example: $\fx(\Proba + \Probb) = \alpha_1 - \beta_1 \cdot ( \Proba + \Probb )$ and $\gy(\Proba) = \beta_2 \cdot \Proba$. That is, both positive and negative network effects change linearly with the respective market shares.  In this example, we can obtain the following uniqueness condition:
$\RL - \alpha_1 - \beta_1 > \beta_2 $.
Namely, if $\RL$ is large enough or $\beta_2$ is small enough, there is a unique Nash equilibrium in MSCG.



Once we obtain the Nash equilibrium $(\Probl^*, \Proba^*)$ of MSCG, we can immediately obtain the Nash equilibrium $(\pl^*, \pa^*)$ of the original PCG by (\ref{eq:price-market-share-rs}).
It is notable that we may not be able to derive the analytical Nash equilibrium of MSCG, as we use the generic functions $\fx(\cdot)$ and $\gy(\cdot)$.
Nevertheless, thanks to the nice property of supermodular game, we can easily numerically compute the Nash equilibrium of MSCG through, for example, the simple best response iteration in \cite{report}.

\section{Layer I -- Commission Bargaining Solution}\label{sec:layer1}



In this section, we study the commission negotiation among the database and the spectrum licensee in Layer I, based on their predictions of the price equilibrium in Layer II and the market equilibrium in Layer III.

\rev{
Specifically, we want to find a feasible revenue sharing percentage $\delta \in [0, 1]$ under RSS that is \emph{satisfactory} for both the database and the spectrum licensee.
}
We formulate the commission negotiation problem as a \emph{one-to-one bargaining}, and study the bargaining solution using the Nash bargaining theory \cite{harsanyi1977bargaining}.

%

Following the Nash bargaining framework, we first derive the database's and the licensee's payoffs when reaching an agreement and when \emph{not} reaching any agreement (hence reaching the disagreement).
Specifically, when reaching an agreement $\delta$, the database's and the licensee's payoffs are  $\Udbrs (\delta)$ and $\Uslrs (\delta)$ derived in
Section \ref{sec:layer2}, respectively.
When not reaching any agreement (reaching the disagreement), the licensee's profit is $\Uslo = 0$, and the database's profit is $\Udbo = \pa^{\dag} \cdot \Proba^{\dag}(\pa^{\dag}) $, where $\pa^{\dag}$ and $\Proba^{\dag}(\pa^{\dag})  $ are the database's optimal price and the corresponding market share in the pure information market.\footnote{Note that such an optimal price and market share can be derived in the same way as in Section \ref{sec:layer2}, by simply setting $\RL = 0$.}
Then, the Nash bargaining solution is formally given by
\begin{equation}\label{eq:NBS-RS}
\begin{aligned}
\max_{\delta \in [0,1]}~ &\left( \Udbrs(\delta) - \Uslo \right) \cdot \left( \Uslrs(\delta) - \Udbo \right) \\
\text{s.~t.~} & \Udbrs(\delta) \geq \Uslo,~~ \Uslrs(\delta) \geq \Udbo.
\end{aligned}
\end{equation}

Note that analytically solving (\ref{eq:NBS-RS}) may be difficult, as it is hard to characterize the analytical forms of $\Udbrs(\delta)$  and $\Uslrs(\delta) $.
Nevertheless, we notice that the bargaining variable $\delta$ lies in closed and bounded range of $[0,1]$, and the objective function of (\ref{eq:NBS-RS}) is bounded.
Hence, there must exist an optimal solution for (\ref{eq:NBS-RS}), which can be found by using many one-dimensional search methods \rev{(e.g., \cite{bargai2001search})}.
%
\begin{proposition}[$\delta$-Bargaining Solution]\label{lemma:existence-NBS}
There must exist an optimal solution for the problem (\ref{eq:NBS-RS}). \rev{If the objective function of (\ref{eq:NBS-RS}) is monotonic, the optimal solution is unique.}
\end{proposition}

\rev{

\begin{figure*}[t]
\centering
\begin{minipage}[t]{0.32\linewidth}
\centering
  \includegraphics[width=2in]{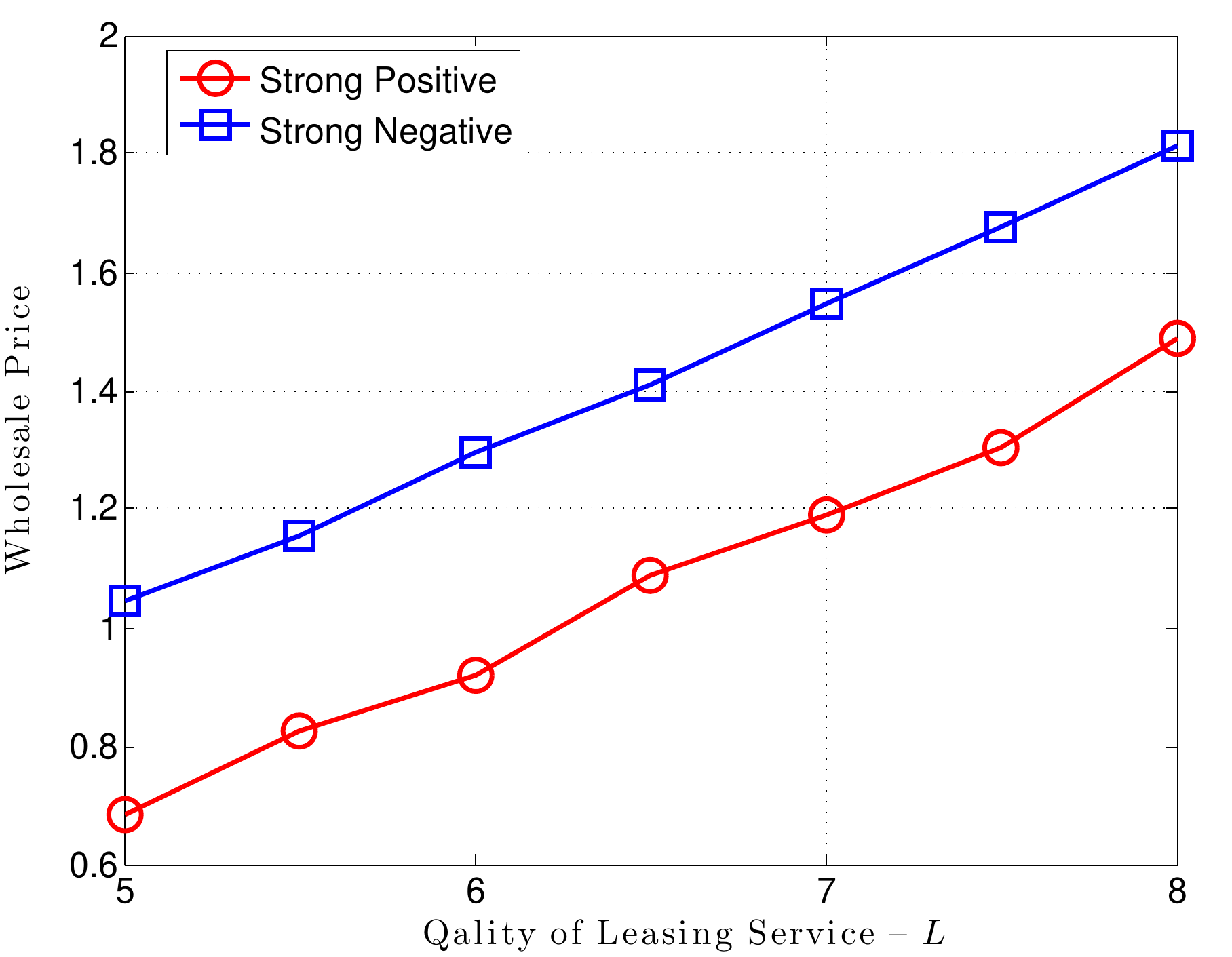}  
  \caption{\rev{Equivalent wholesale price vs $\RL$ under strong positive network externality and strong negative network externality}}.\label{fig:price-vs-diff-L}
\end{minipage}
\begin{minipage}[t]{0.32\linewidth}
\centering
  \includegraphics[width=2in]{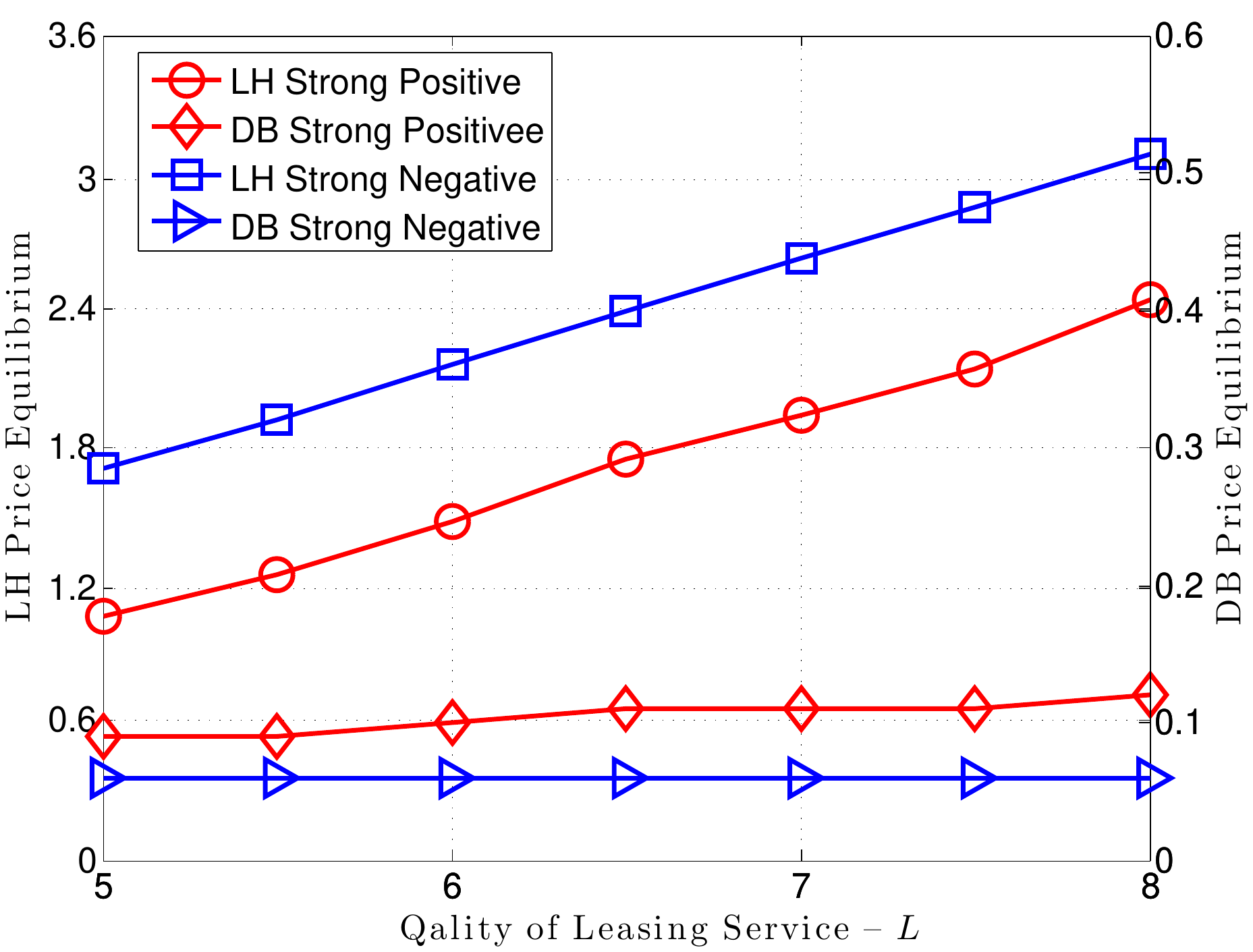}
  \caption{\rev{The \db's and the \lh's retail price (where LH denotes Licensee and DB denotes database) vs $\RL$ under different network externality}}.\label{fig:price-vs-diff-L2}
\end{minipage}
\begin{minipage}[t]{0.32\linewidth}
\centering
  \includegraphics[width=2.5in]{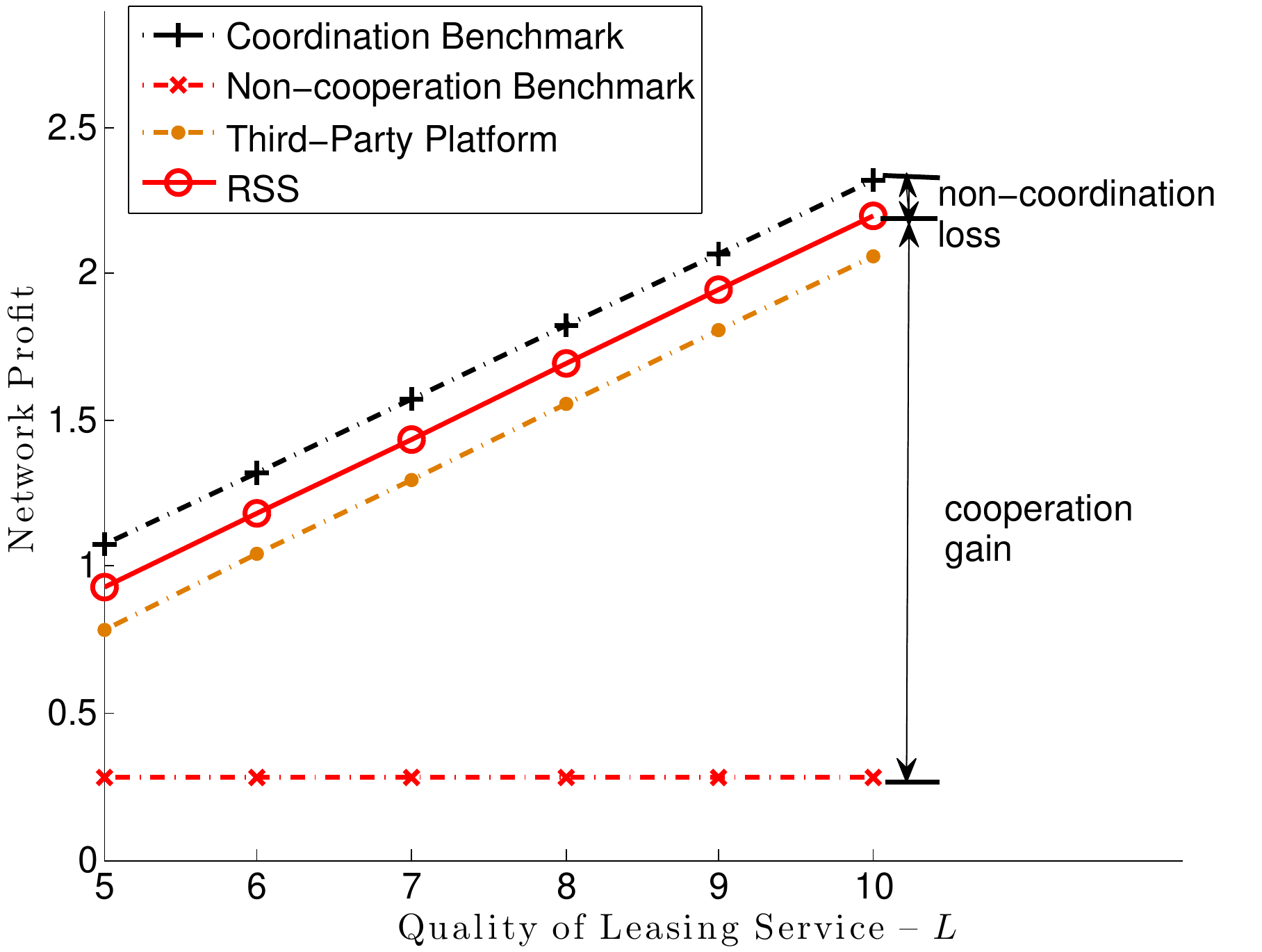}
  \caption{\rev{Network profit vs $\RL$ under both revenue sharing scheme (RSS) under fixed network externality}.}\label{fig:NP-vs-diff-L}
\end{minipage}
\vspace{-4mm}
\end{figure*}

\section{Simulation Result}\label{sec:simulation}

In this section,
we provide simulations to evaluate the system performance (e.g., the network profit, the \db's profit, and the \lh's profit) achieved under revenue sharing scheme (RSS).

\subsubsection{Commission Bargaining Solution and Equilibrium Retail Price}
We first show the Nash bargaining solution of RSS, and the corresponding equilibrium retail prices.
In this simulation:
We choose the function $\fx( \Probl ) = \alpha_1 - {\beta_1}\cdot ( 1 - \Probl )^{\gamma_1}$ to model the negative network externality, and choose the function $\gy = \alpha_2 + (\beta_2 - \alpha_2) \cdot {\Proba}^{\gamma_2}$ to model the positive network externality.
We fix $\alpha_1 = 1.8$, $\alpha_2 = 1$, $\beta_1 = 0.8$, $\beta_2 = 1.2$, $\gamma_1 = 0.8$, $\gamma_2 = 0.6$, and change the leasing service quality $\RL$ from $6$ to $10$.

%
%

Figure \ref{fig:price-vs-diff-L} shows the bargaining solution under different degree of network externality.
For fair comparison, we transform the revenue sharing factor $\delta^{*}$ under RSS into the equivalent wholesale price by $\w^{*} = \delta^{*} \cdot \Probl^{*}$.
From this figure, we can see that the equivalent wholesale price increases with $\RL$ under different degree of network externality.
This is because a higher quality of leasing service attracts more \eus~to the licensed channels, and hence a higher wholesale price is desired to compensate the database's  revenue loss from the advanced service.
Moreover, the wholesale price under strong negative networker externality is higher than that under strong positive network externality.
\rev{This is because when the information market is negative network externality, the increases number of {\eus}  will severely jeopardize the quality of unlicensed TV channels. Hence, {\eus} are willing to choose licensed channels in order to obtain guarantee quality of service. In such case, the database wants to increase the wholesale price to maintain its revenue.}

Figure \ref{fig:price-vs-diff-L2} shows the equilibrium retail prices under the bargaining solution.
We can see that the equilibrium price of the \db~(denoted by DB) is almost independent of $\RL$, while the equilibrium   price of the \lh~(denote by LH) increases with $\RL$.
This is because a higher quality of leasing service will attract more \eus~to the licensed channels, and thus allows the \lh~to charge a higher service price.
Moreover, the equilibrium price of \lh~under strong negative network externality is higher than that under strong positive externality.
This is because under strong negative network externality, the small increase of {\eus} in the information market will dramatically decreases the quality of unlicensed TV channels. Hence, the licensee can charge a higher price due to more {\eus} will choose leasing service.


\subsubsection{System Performance}
Now we show the network profit, i.e., the aggregate profit of the \db~and the \lh~achieved under both RSS in Figure \ref{fig:NP-vs-diff-L} given the fixed network externality.
In this figure, we use the black dash-dot line (with mark $+$) to denote the coordination benchmark, where the \db~and the \lh~act as an integrated party to maximize their aggregate profit. We use the red dash-dot line (with mark $\times$) to denote the non-cooperation benchmark (with pure information market only), where the \db~does not want to display the \lh's licensed \ch~information. The brown dash-dot line (with mark $\bullet$) denotes the case where the \lh~sells channels on a third-party spectrum market platform.

Figure \ref{fig:NP-vs-diff-L} shows that our proposed RSS outperform the non-cooperation scheme significantly (e.g., increasing the network profit up to $87\%$).
It further shows that the gap between our proposed RSS and coordination benchmark is small (e.g., less than $4\%$).
Such a gap is caused by the \emph{imperfect} coordination of the database and the licensee.
In other words, they cooperate somewhat  but not coordinate completely.
We call this gap as the non-coordination loss.
We can also see that our proposed RSS always outperforms the scheme with a third-party platform, in terms of the network profit.

}

\section{Conclusion}\label{sec:con}
\rev{
In this paper, we proposed a database-provided hybrid spectrum and information market, and analyze the interactions among the database, the licensee, and the unlicensed users systematically.
We also analyze how the network externalities (of the information market) affect these interactions.
Our work not only captures the performance gain introduced by the hybrid market, but also characterizes the impact of different degree of networker externality on the market equilibrium behaviours of all parties involved.
There are several possible directions to extend this work.
One is to consider an oligopoly scenario with multiple databases (hence multiple platforms).
In this
scenario, databases compete with each other for unlicensed users as well as for spectrum licensees.
}


\appendix
\section{Appendix}\label{sec:appendix}

\section{Appendix}\label{sec:appendix}

\subsection{Property of Information Market}
In this section, we will discuss the properties of two types of network externalities in the information market: the negative network externality and the positive network externality.
For illustration purpose, we explicit define the advance information as those proposed by Luo \emph{et al.} in \cite{luo2014wiopt,luo2014SDP}. In the following, we first define the advanced information as the interference level on each channel, then we characterize the information value to the users. Based on that, we can further characterize the properties of the information market.

\subsubsection{Interference Information}
\label{sec:interference_info}
For each {\eu} $n\in \Nset$ operating on the TV \ch, each channel $\k$ is associated with an \emph{interference level}, denoted by  $\InfTot_{\n,\k}$, which reflects the aggregate interference from all other nearby devices (including TV stations
and other {\eus}) operating on this channel.
Due to the fast changing of wireless channels and the uncertainty of {\eus}' mobilities and activities, the interference $\InfTot_{\n,\k}$ is a random variable. We impose assumptions on the interference $\InfTot_{\n,\k}$ as follows.
\begin{assumption}\label{assum:iid}
For each {\eu} $n\in \Nset$, each channel $\k$'s interference level $\InfTot_{\n,\k}$ is \emph{temporal-independence} and \emph{frequency-independence}.
\end{assumption}

This assumption shows that
(i) the interference $\InfTot_{\n,\k}$ on channel $\k$ is independent identically distributed (iid) at different times, and (ii)
the interferences on different channels, $\InfTot_{\n,\k}, \k\in\Kset$, are also iid at the same time.\footnotesc{\rev{Note that the iid assumption is a reasonable approximation of the practical scenario. This is because {\eus} with basic service will randomly choose one TV \ch, hence the number of such {\eus} per channel will follow the same distribution. For {\eus} with advanced service, they will go to the TV \ch~with the minimum realized interference. If the interference among each pair of users is iid over time, then statistically the number of such users in each channel will also follow the same distribution. Note that even though all channel quality distributions are the same, the realized instant qualities of different channels are different. Hence, the advanced information provided by the \db~is still valuable as such an advanced information is accurate interference information.} }
As we are talking about a general \eu~$n$, \textbf{we will omit the {\eu} index $n$ in the notations (e.g., write $\InfTot_{\n,\k}$ as $\InfTot_{\k}$), whenever there is no confusion caused.}
Let $H_{\InfTot}(\cdot)$ and $h_{\InfTot}(\cdot)$ denote the cumulative distribution function (CDF) and probability distribution function (PDF) of $\InfTot_{\k}$, $\forall \k\in\Kset$.\footnotesc{In this paper, we will conventionally  use $H_X(\cdot)$ and $h_X(\cdot)$ to denote the CDF and PDF of a random variable $X$, respectively.}~~~~

Usually, a particular {\eu}'s experienced interference $\InfTot_{\k}$ on a \ch~$\k$ consistss of the following three components:
\begin{enumerate}
\item
$\InfTV_{\k}$: the interference from licensed TV stations;
\item
$\InfEU_{\k,m}$: the interference from another {\eu} $m$ operating on the same channel $k$;
\item
$\InfOut_{\k}$: any other interference from outside systems.
\end{enumerate}
The total interference on channel $k$ is
$
\InfTot_{\k} = \InfTV_k + \InfEU_k + \InfOut_k
$, where  $\InfEU_k  \triangleq \sum_{m \in \Nkset} \InfEU_{\k,m}$ is the total interference from all other {\eus} operating on channel $k$ (denoted by $\Nkset$).
Similar to $\InfTot_{\k}$, we also assume that  $ \InfTV_k , \InfEU_k , \InfEU_{\k,m}$, and $ \InfOut_k$ are random variables with \emph{temporal-independence} (i.e., iid across time) and \emph{frequency-independence} (i.e., iid across frequency).
We further assume that $\InfEU_{\k,m}$ is \emph{user-independence}, i.e., $\InfEU_{\k,m}, m\in \Nkset$, are iid.
{It is important to note that \textbf{different {\eus} may experience different interferences $\InfTV_k$ (from TV stations), $ \InfEU_{\k,m} $ (from another \eu~operating on the same \ch), and $ \InfOut_k$ (from outside systems) on a channel $k$, as we have omitted the \eu~index $n$ for all these notations for clarity.}}

Next we discuss these interferences in more details.

\begin{itemize}
\item
The \db~is able to compute the interference $\InfTV_k$ from TV stations to every {\eu} (on channel $k$), as it knows the locations and channel occupancies of all TV stations.~~~~

\item
The \db~cannot compute the interference
$\InfOut_k$ from outside systems, due to the lack of outside interference source information.
Thus, the interference $\InfOut_{\k}$ will \emph{not} be included in a database's advanced information sold to {\eus}.~~~~~~~~

\item
The \db~may or may not be able to compute the interference $\InfEU_{\k,m}$ from another {\eu} $m$, depending on whether {\eu} $m$ subscribes to the \db's advanced service.
Specifically, if {\eu} $m$ subscribes to the advanced service, the \db~can predict its channel selection (since the {\eu} is fully rational and will always choose the channel with the lowest interference level indicated by the \db~at the time of subscription), and thus can compute its interference to any other \eu.
However, if {\eu} $m$ only chooses the \db's basic service, the \db~cannot predict its channel selection, and thus cannot compute its interference to other \eus.~~~~
\end{itemize}

For convenience, we denote $\Nkset^{[\A]}$, as the set of {\eus} operating on \ch~$k$ and subscribing to the \db's advance service (i.e., those choosing the strategy $\l = \A$), and $\Nkset^{[\B]}$ as the set of {\eus} operating on channel $k$ and choosing the \db's basic service (i.e., those choosing the strategy $\l = \B$).
That is, $\Nkset^{[\A]} \bigcup \Nkset^{[\B]} = \Nkset $.
Then, for a particular {\eu}, its experienced interference (on channel $k$) \textbf{known by the \db} is
\begin{equation}\label{eq:known_inf}
\begin{aligned}
\textstyle
\InfKnown_{\k} \triangleq \InfTV_{\k} + \sum_{m \in \Nkset^{[\A]}} \InfEU_{\k,m},
\end{aligned}
\end{equation}
which contains the interference from TV licensees and all {\eus} (operating on channel $k$)  subscribing to the \db's advanced service.
The \eu's experienced interference (on channel $k$) \textbf{\emph{not} known by the \db} is
\begin{equation}\label{eq:unknown_inf}
\begin{aligned}
\textstyle
\InfUnknown_{\k}  \triangleq \InfOut_{\k} + \sum_{m \in \Nkset^{[\B]}} \InfEU_{\k,m},
\end{aligned}
\end{equation}
which contains the interference from outside systems and all {\eus} (operating on channel $k$) choosing the \db's basic service.
Obviously, both $\InfUnknown_{\k}$ and $\InfKnown_{\k}$ are also random variables with temporal- and frequency-independence.
Accordingly, the total interference on {\ch} $\k$ for a {\eu} can be written as
$\textstyle
\InfTot_{\k} = \InfKnown_{\k} + \InfUnknown_{\k}.$~~~~~~

\textbf{Since the \db~knows only $\InfKnown_{\k}$, it will provide this information (instead of the total interference $\InfTot_k$) as the {advanced service} to a subscribing {\eu}.}
It is easy to see that the more {\eus} subscribing to the \db's advanced service, the more information the \db~knows, and the more accurate the \db~information will be.

Next we can characterize the accuracy of a database's information explicitly. Note that $\Proba$ and $\Probl$ denote the fraction of {\eus} choosing the advanced service and leasing licensed spectrum, respectively. Moreover, $(1 - \Proba - \Probl)$ denotes the fraction of {\eus} choosing the basic service.
Hence, there are $(1 - \Probl) \cdot \N$ {\eus} in the network that we consider operating on the TV \chs.
Due to the Assumption \ref{assum:iid}, it is reasonable to assume that each channel $k\in\Kset$ will be occupied by an average of $\frac{\N}{\K} \cdot ( 1 - \Probl)$ {\eus}.
Then, among all $\frac{\N}{\K} \cdot ( 1 - \Probl)$ {\eus} operating on channel $k$, there are, \emph{on average}, $\frac{\N}{\K}\cdot\Proba$ {\eus}  subscribing to the \db's advanced service, and $\frac{\N}{\K}\cdot ( 1 - \Proba - \Probl )$ {\eus} choosing the \db's basic service.
That is, $| \Nkset | = \frac{\N}{\K} \cdot ( 1 - \Probl)$, $| \Nkset^{[\A]} | = \frac{\N}{\K}\cdot\Proba$, and $|\Nkset^{[\B]} | = \frac{\N}{\K}\cdot(1-\Proba-\Probl)$.\footnotesc{{Note that the above discussion is from the aspect of expectation, and in a particular time period, the realized numbers of {\eus} in different channels may be different.}}
Finally, by the {user-independence} of $\InfEU_{\k,m}$, we can immediately calculate the distributions of $ \InfKnown_{\k}$ and $ \InfUnknown_{\k}$ under any given market share $\Proba$ and $\Probl$ via (\ref{eq:known_inf}) and (\ref{eq:unknown_inf}).

\subsubsection{Information Value}
Now we evaluate the value of the \db's advanced information to {\eus}, which is
reflected by the {\eu}'s benefit (utility) that can be achieved from this information.

We first consider the expected  utility of a {\eu} when choosing the \db's basic service (i.e., $\l=\B$).
In this case, the {\eu} will randomly choosing a \tvch~based on the information provided in the free basic service,
and its expected data rate is
\begin{equation}\label{eq:rate-random-fixed}
\begin{aligned}
\textstyle
\R_0( 1 - \Probl)  \textstyle = \Ex_{Z} [\rt(\InfTot)] = \int_{z} \rt(z) \mathrm{d} H_{\InfTot}(z), \\
\end{aligned}
\end{equation}
where $\rt(\cdot)$ is the transmission rate function (e.g.,
the Shannon capacity) under any given interference.
As shown in Section \ref{sec:interference_info}, each channel $k\in\Kset$ will be occupied by an average of $\frac{\N}{\K} \cdot ( 1 - \Probl)$ {\eus} based on the Assumption \ref{assum:iid}. Hence,
$\R_0( 1 - \Probl) $ depends only on the distribution of the total interference $\InfTot_k$, and thus depends on the fraction of {\eus} operating on TV \chs~(i.e., $1 - \Probl$). Then the expected utility provided by the basic service is:
\begin{equation}\label{eq:utility-random-fixed}
\begin{aligned}
\textstyle
\RB(1 - \Probl) = \textstyle \ut\bigg(\R_0( 1 - \Probl)\bigg),
\end{aligned}
\end{equation}
where $\ut(\cdot)$ is the utility function of the \eu. We can easily check that
the more {\eus} operating on the TV \chs, the higher value of $\InfTot_k$ is, and thus the lower expected utility provided by the basic service. Hence, the basic service's expected utility reflects the congestion level of the TV \chs.
We use the function $\fx(\cdot)$ to characterize the congestion effect and have $\fx(1 - \Probl) = \RB( 1 - \Probl)$.

Then we consider the expected utility of a {\eu} when subscribing to the \db's advance service.
In this case, the {\db} returns the interference $\{\InfKnown_{k}\}_{k\in\Kset}$ to the {\eu} subscribing to the advanced service, together with the basic information such as the available channel list.
For a rational {\eu}, it will always choose the channel with the minimum $\InfKnown_{k}$ (since $\{\InfUnknown_{k}\}_{k\in\Kset}$ are iid).
Let $\InfKnownMin^{[l]} =  \min\{ \InfKnown_{1},  \ldots, \InfKnown_{K} \} $ denote the minimum interference indicated by the \db's advanced information.
Then, the actual interference experienced by a {\eu} (subscribing to the \db's advanced service) can be formulated as the sum of two random variables, denoted by $\InfTotA = \InfKnownMin + \InfUnknown$. Accordingly, the {\eu}'s expected data rate under the strategy $\l = \A$ can be computed by
\begin{equation}\label{eq:rate-pay-fixed}
\begin{aligned}
\Ra(\Proba, \Probl) & \textstyle = \Ex_{\InfTotA} \big[ \rt \left( \InfTotA \right) \big] = \int_z \rt(z)  \mathrm{d} H_{\InfTotA}(z),
\end{aligned}
\end{equation}
where
$H_{\InfTotA}(z)$ is the CDF of $\InfTotA$. It is easy to see that $\Ra$ depends on the distributions of $\InfKnown_k$ and $\InfUnknown_k$, and thus depend on the market share $\Probl$ and $\Proba$.
Thus, we will write $\Ra$ as $\Ra(\Probl,\Proba)$. Accordingly, the advanced service's utility is:
 \begin{equation}\label{eq:utility-pay-fixed}
\begin{aligned}
\RA(\Probl,\Proba) & \textstyle  \triangleq \ut\bigg( \Ra( \Probl,\Proba ) \bigg)
\end{aligned}
\end{equation}

Note that the congestion effect also affects the value of $\RA$. However, compared with the utility of {\eu} choosing basic service, the benefit of a {\eu} subscribing to the \db's advanced information is coming from the $\InfKnownMin^{[l]}$, i.e., the minimum interference indicated by the \db's advanced information. As the value of $\InfKnownMin^{[l]}$ depends on the $\Proba$ only, we can get the approximation $\RA = \fx(1 - \Probl) + \gy(\Proba)$, where function $\gy(\cdot)$ characterize the benefit brought by $\InfKnownMin^{[l]}$ and denotes the positive network effect.



By further checking the properties of $\RB(1 - \Probl)$ and $\RA(\Probl,\Proba)$, we have the Assumption \ref{assum:congestion} and Assumption \ref{assum:positive}.

%

\end{document}